\documentclass[a4paper,runningheads,envcountsamme]{article}
  
\usepackage[inline]{enumitem}
\usepackage{thmtools}
\usepackage{thm-restate}
  
  \usepackage[utf8]{inputenc}
\usepackage[T1]{fontenc}
\usepackage{amsmath,amsthm,amssymb,amsfonts,mathtools}
\usepackage{xspace}
\usepackage[ruled,vlined,linesnumbered]{algorithm2e}
\usepackage{hyperref}
\usepackage{xcolor}
\hypersetup{
    colorlinks=true,
    linkcolor=blue,
    urlcolor=blue,
    citecolor=blue
}
\usepackage{cleveref,soul}
\usepackage{verbatim}
\usepackage{tikz}
\usepackage{setspace}
\usetikzlibrary{decorations}
\usepackage{comment}
\usepackage{graphicx}
\usepackage{subcaption}
\usepackage{caption}
\usepackage{authblk}
\usepackage{float}

\usetikzlibrary{positioning, fadings, backgrounds}

\usepackage[textsize=footnotesize,color=green!40]{todonotes}
\usepackage[hmargin=2.5cm,vmargin=3cm]{geometry}

\tikzset{snake it/.style={decorate, decoration=snake}}
\newcommand{\Ball}[2]{B_{#1}\left({#2}\right)}
\newcommand{\Sphere}[2]{N_{#1}(#2)}

\newcommand{\InducedSub}[2]{#1\left[#2\right]}

\newcommand{\dist}[2]{\mathsf{d}\left(#1,#2\right)}
\newcommand{\distG}[3]{\mathsf{d}_{#1}\left(#2,#3\right)}\newcommand{\weight}[1]{\mathsf{w}\left(#1\right)}
\newcommand{\Interval}[2]{\left[#1,#2\right]}
\newcommand{\ecc}[1]{ecc\left(#1\right)}
\newcommand{\cardinal}[1]{\left|#1\right|}

\newcommand{\Intervalpath}[2]{F_{#1,#2}}
\newcommand{\Projection}[2]{Pr\left(#1,#2\right)}

\definecolor{dartmouthgreen}{rgb}{0.05, 0.5, 0.06}

\newtheorem{theorem}{Theorem}

\newtheorem{lemma}[theorem]{Lemma}
\newtheorem{observation}[theorem]{Observation}

\newtheorem{claim}{Claim}[theorem]

\newenvironment{claimproof}[1][\proofname]{%
  \begin{proof}[#1]
}{%
  \end{proof}%
}

\title{ Minimum eccentricity shortest paths of $K_{2,3}$-minor-free graphs}

\begin{document}

\author[1]{Dibyayan Chakraborty}
\author[2]{Sandip Das}
\author[2]{Sk Samim Islam}
\author[2]{Ritam Manna Mitra}
\author[2]{Saumya Sen}
\affil[1]{University of Leeds, Leeds, UK}
\affil[2]{Indian Statistical Institute, Kolkata, India}




\maketitle

\begin{abstract}
Given a simple, undirected, and unweighted graph $G$, and an integer $R$, the objective of the \textsc{Minimum Eccentricity Shortest Path (MESP)} is to decide whether there exists an \emph{isometric path} $P$ in $G$ such that the distance from every vertex in the graph to its nearest vertex in $P$ is at most $R$. In this paper, we prove that MESP admits an $O(n^4)$-time algorithm on $K_{2,3}$-minor-free graphs. Our algorithm has a cubic running time when the inputs are restricted to a cactus.  
\end{abstract}


\section{Introduction}\label{sec:intro}

All graphs considered in this paper are simple and undirected. Unless otherwise stated, a graph is also unweighted. In this paper, we study the computational complexity of  \textsc{Minimum Eccentricity Shortest Path (MESP)}, a ``distance-based '' algorithmic problem introduced by
Dragan and Leitert~\cite{dragan2017minimum}. Given a graph $G$ the objective of \textsc{MESP} is to find an \emph{isometric} path (i.e. a path of minimal length between its end vertices) $P$ such that the maximum distance between any vertex of $G$ and $P$ is the minimum. The decision version of the problem is stated as: Given a graph $G$ and an integer $R$, the objective of \textsc{MESP} is to decide if $G$ has an isometric path $P$ such that the maximum distance between any vertex of $G$ and $P$ is at most $R$. \textsc{MESP} is related to fundamental algorithmic problems like \emph{minimum distortion embeddings}, \emph{$k$-dominating sets} and the \emph{$k$-laminar problem}~\cite{dragan2017minimum}. From practical perspectives, \textsc{MESP} can be thought of as the problem of finding a ``short'' route that is also highly accessible in a network.

Dragan and Leitert~\cite{dragan2017minimum} proved that, for any integer $k$, it is possible to decide in $O(n^{2k+2})$-time,  if a simple, undirected, and unweighted graph has an isometric path of eccentricity at most $k$. On the other hand, \textsc{MESP} remains NP-complete even on \emph{subcubic partial grids}~\cite{chakraborty2025additive}, a subclass of planar bipartite graphs of maximum degree at most $3$. However, the computational complexity of \textsc{MESP} on important restricted classes of planar graphs like \emph{outerplanar} graphs, \emph{series-parallel} graphs, \emph{solid grids} etc., remain unknown. 

In this paper, we prove that \textsc{MESP} admits a polynomial time algorithm on $K_{2,3}$-minor-free graphs, a superclass of outerplanar graphs. Specifically, we prove the following theorem.

\begin{restatable}{theorem}{main}\label{thm:main}
	\textsc{MESP} admits an $O(n^4)$-time algorithm on $K_{2,3}$-minor-free graphs, where $n$ is the number of vertices of the input graph.
\end{restatable}

In the above theorem, $K_{2,3}$ denotes the complete bipartite graph with partite sets containing two and three vertices respectively. A graph $H$ is a minor of a graph $G$ if $H$ can be obtained by deleting vertices, deleting edges or contracting edges of $G$. Clearly, $K_{2,3}$-minor-free graphs are planar. Moreover, $K_{2,3}$-minor-free graphs have \emph{treewidth} at most $3$. 
Note that the computational complexity of \textsc{MESP} on graphs with constant \emph{treewidth} is unknown. In fact, it is unknown if MESP admits a polynomial time algorithm on graphs of treewidth at most $2$, which are equivalent to \emph{series-parallel} graphs. Therefore, even though the $K_{2,3}$-minor-free graphs have treewidth at most $3$, the existing results are not enough to derive \Cref{thm:main}. 
A graph $G$ is a \emph{cactus} if every biconnected component of $G$ is a cycle. We show that if the inputs are restricted to \emph{cactus}, the time complexity of our algorithm can be reduced to $O(n^3)$.

\begin{restatable}{theorem}{mainn}\label{thm:mainn}
		\textsc{MESP} admits an $O(n^3)$-time algorithm on cactus with $n$ vertices.
\end{restatable}

\paragraph*{Overview of our algorithm.} Given a $K_{2,3}$-minor-free graph $G$, our algorithm iterates over every pair of vertices $s,t$ and finds an isometric path $P_{s,t}$ with $s,t$ as end-vertices that minimizes the eccentricity and reports the one with the overall minimum eccentricity. 

For a fixed pair of vertices $s,t\in V(G)$, we carefully  analyse the structure of \emph{$(s,t)$-interval}, i.e., the subgraph created by the set of edges lying on some isometric path between $s$ and $t$. Even though there could be exponentially many distinct $(s,t)$-isometric paths,
we observe that the vertices of the $(s,t)$-interval can be decomposed into a family $\mathcal{Q}$ consisting of $O(n)$ isometric paths. We use these properties to create an edge weighted auxiliary directed graph $\overrightarrow{H}$ that contains two special vertices $s',t'$ such that any $(s',t')$-path in $\overrightarrow{H}$ corresponds to an $(s,t)$-isometric path in $G$ and vice versa. Moreover, we ensure that the maximum weight of an $(s',t')$-path in $\overrightarrow{H}$ is the eccentricity of the corresponding $(s,t)$-shortest path in $G$. Finally, we use standard text book algorithms to find a $(s',t')$-path in $\overrightarrow{H}$ that minimizes the maximum edge weight in the path.

We note that our technique is different from those used in the literature to design polynomial-time algorithms for MESP. For example, finding a diametral pair and any isometric path between them is a popular technique that has been used to solve MESP on trees, distance hereditary graphs etc. However, such a method will not work even on cactus as exhibited in the following figure. 

\begin{figure}[ht]
    \centering
    \includegraphics[scale=0.5]{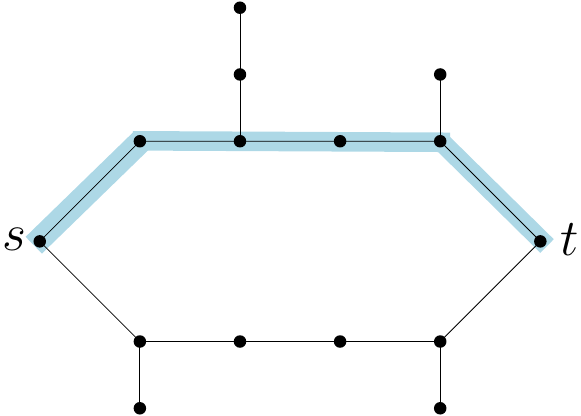}
    \caption{An example of a cactus with no diametral path having the smallest eccentricity. A minimum eccentricity shortest path is shaded in blue.}
    \label{fig:cac}
\end{figure}

\paragraph*{Related works}

Clearly, a graph has an isometric path of eccentricity $0$ if and only if it is a path. The class of graphs that has an isometric path of eccentricity at most $1$ (also known as a \emph{dominating shortest path}), is richer and includes \emph{asteroidal triple-free} graphs~\cite{corneil1999linear} (and thus \emph{interval} graphs, \emph{cocomparability} graphs, etc.) Convex bipartite graphs have an isometric path of eccentricity at most $2$~\cite{dragan2017line}. On the other hand, the eccentricity of isometric paths of trees is unbounded.

Dragan \& Leitert~\cite{JGAA-394} studied the computational complexity of \textsc{MESP} on structured graph classes. They designed polynomial time algorithms for \textsc{MESP} on \emph{distance-hereditary} graphs, \emph{chordal} graphs etc. More generally, they proved that \textsc{MESP} admits an $O(n^{4\delta+4})$-algorithm on graphs with bounded \emph{hyperbolicity} $\delta$, a graph invariant introduced by Gromov~\cite{gromov1987} which measures how far the  distance metric of a graph is from that of a tree. In fact, the above is a corollary of the following result proved by the same authors: \textsc{MESP} admits a polynomial time algorithm on graphs with bounded \emph{projection gap}. We do not define \emph{projection gap} in this paper but cycles have unbounded projection gap. Therefore, \Cref{thm:main,thm:mainn} cannot be derived from the above result.

Kučera \& Suchý~\cite{kuvcera2023minimum} studied the fixed-parameter tractability of MESP. They proved that MESP admits FPT algorithms w.r.t. modular width, distance to cluster graph, the combination of treewidth with the desired eccentricity, and maximum leaf number.  Bhyravarapu et al.~\cite{bhyravarapu2023parameterized} extended some of the above results by providing FPT algorithms for MESP parameterized by the disjoint paths deletion set, split vertex deletion set, or the sum of feedback vertex set number
and the desired eccentricity. On the negative side, MESP is W[2]-hard w.r.t. the desired eccentricity~\cite{dragan2017minimum} and chordal vertex deletion set~\cite{bhyravarapu2023parameterized}. Each of the above parameters remains unbounded in outerplanar graphs.

The approximability of \textsc{MESP} has been extensively studied. Dragan \& Leitert~\cite{dragan2017minimum} developed a $2$-approximation, a $3$-approximation, and an $8$-approximation algorithm that run in $O(n^3)$-time, $O(nm)$-time, and $O(m)$-time, respectively, where $n$ and $m$ denote the numbers of vertices and edges of the input graph respectively. Birmelé et al.~\cite{birmele2016minimum} proposed a $3$-approximation algorithm for MESP which runs in linear time. In the PhD Thesis of A.O. Mohammed~\cite{mohammed2019slimness}, an $O(\delta)$-approximation algorithm for MESP on $\delta$-hyperbolic graphs has been proposed. Such algorithms do not provide even a constant-factor approximation algorithms on outerplanar graphs.

Several generalizations of MESP have been studied. Birmelé et al.~\cite{birmele2016minimum} studied \emph{hub-laminar decompositions} that aims to decompose a graph into subgraphs containing isometric paths of bounded eccentricity. Chakraborty \& Vaxès~\cite{chakraborty2025additive} studied the approximability of \textsc{$k$-Geodesic Center} where the objective is to select $k$ isometric paths of the input graph that minimize the maximum eccentricity. The authors proposed an additive $O(\delta)$-approximation algorithm for $\delta$-hyperbolic graphs. 

\section{Preliminaries}\label{sec:prelim}

 Let $G$ be a graph. For a set $S$ of vertices, $\InducedSub{G}{S}$ denotes the subgraph of $G$ induced by $S$, and $G-S =\InducedSub{G}{V(G)\setminus S}$. When $S=\{w\}$, we write $G-\{w\}$ as $G-w$. For two graphs $H$ and $G$, we write $H\subseteq G$ if $H$ is a subgraph of $G$. A set $S$ of vertices is a \emph{cutset}, if $G-S$ is disconnected. A set $S$ of vertices is an \emph{$(u,v)$-cutset}, if $u$ and $v$ lie in two different components of $G-S$. A block of a graph G is a maximal connected subgraph of $G$ that has no cut-vertex. 
 A \emph{chord} of a cycle $C$ is an edge not in the edge set of $C$ but whose endpoints lie in the vertex set of $C$.  A \emph{hole} is a chordless cycle of length at least $4$.

  The length of a path $P$ is the number of edges in $P$. For two vertices $u,v\in V(G)$, a $(u,v)$-path (resp. $(u,v)$-induced path) is a path (resp. an induced path) between $u$ and $v$. An $(u,v)$-isometric path is an $(u,v)$-path with smallest possible length. 
	 The \emph{distance} between two vertices $u,v\in G$, denoted by $\dist{u}{v}$, is the length of an $(u,v)$-isometric path. For a set $S$ and a vertex $u$, the distance between $u$ and $S$, denoted by $\dist{u}{S}=\min\{\dist{u}{v}\colon v\in S\}$.	For an integer $k\geq 0, r\in V(G)$, let $\Ball{k}{r}$ denote the \emph{ball of radius $k$} centered at $r$, i.e., $\Ball{k}{r} = \{u\in G\colon \dist{r}{u}\leq k \}$. Similarly let $\Sphere{k}{r}$ denote the $k^{th}$ neighbourhood of $r$, i.e.,  $\Sphere{k}{r} = \{u\in G\colon \dist{r}{u}= k \}$.  We shall write $\Sphere{1}{r}$ simply as  $N(r)$. For $H \subseteq G$, $\Sphere{H}{x}$ is the set of neighbours of $x$ in subgraph $H$ of $G$. 
    

    The \emph{eccentricity} of a vertex $r$, denoted by $\ecc{r}$ is the largest integer $k$ such that $\Sphere{k}{r}$ is non-empty. The eccentricity of a set $S\subseteq V(G)$ is $\ecc{S}=\max\{\dist{v}{S}\colon v\in V(G)\}$. For two vertices $u,v\in V(G)$, the \emph{interval of $u,v$}, denoted by $\Interval{u}{v}$, is the set of all vertices that lie in some $(u,v)$-isometric path.  Let $\Intervalpath{u}{v}$ denote the subgraph of $G$ whose edge set is the union of all $(u,v)$-isometric paths and whose vertex set is $\Interval{u}{v}$. 

    \section{The algorithm}\label{sec:algorithm}

For the entirety of this section, $G$ shall denote a fixed $K_{2,3}$-minor-free graph. 
\subsection{Construction of the auxiliary graph}

First we state the main ideas of constructing the auxiliary graph. For a pair of vertices $s,t\in V(G)$, we analyse the subgraph $\Intervalpath{s}{t}$ and observe the following three facts: (a)~The set of all induced cycles of $\Intervalpath{s}{t}$ admits a linear ordering (see \Cref{lem:linear-order}). 
(b)~The projection of any vertex $u\not\in \Interval{s}{t}$ on $\Interval{s}{t}$ is restricted to at most four vertices that lie in the same induced cycle or in two induced cycles that are consecutive according to the above linear order (see \Cref{lem:pro2} and \Cref{lem:pro4}). 
(c)~Intersection of any $(s,t)$-isometric path with an  induced cycle creates at most five maximal isometric subpaths (see \Cref{lem:4-paths}).


The above properties allow us to create a directed acyclic graph $\overrightarrow{H}$ whose vertices correspond to the important subpaths of induced cycles of $\Intervalpath{s}{t}$. The directed edges of $\overrightarrow{H}$ are assigned between vertices that correspond to subpaths of two induced cycles that are consecutive in the linear order. For technical reasons we add two dummy vertices $s',t'$ in $\overrightarrow{H}$ which will be the \emph{source} and \emph{sink}, respectively. Finally, we show that there is a correspondence between the set of all $(s,t)$-isometric paths in $G$ and the set of all $(s',t')$-paths in $\overrightarrow{H}$ (see \Cref{lem:path-correspond}). Moreover, we shall show that the size of $V(\overrightarrow{H})$ is $O(n)$. We formalise the above ideas below. First, we prove some properties of $K_{2,3}$-minor-free graphs.

\begin{lemma}\label{lem:same-block}
	For two vertices $s,t\in V(G)$, let $u,v$ be two vertices of $H\coloneqq \Intervalpath{s}{t}$ such that $\distG{H}{s}{u}=\distG{H}{s}{v}$. Then, $u,v$ lie in the same block of $H$. 
\end{lemma}
\begin{proof}
	   Clearly, $\cardinal{\Sphere{\distG{H}{s}{u}}{s}} >1$. Let $i$ be the maximum integer less than $\distG{H}{s}{u}$ such that $\cardinal{\Sphere{i}{s}} =1$ and $j$ be the minimum integer greater than $\distG{H}{s}{u}$ such that $\cardinal{\Sphere{j}{s}} =1$. Let $\Sphere{i}{s} = \{x \}$ and $\Sphere{j}{s} = \{y \}$. Thus from the definition there exists a block, say $C$ such that $x,y \in V(C)$. Thus $u,v \in V(C)$. 
\end{proof}




\begin{lemma}\label{lem:2in1level}
	For any two vertices $s,t\in V(G)$ and integer $k$ such that $1 \leq k < \dist{s}{t}$, $\cardinal{V(\Intervalpath{s}{t})\cap \Sphere{k}{s}}\leq 2$. 
\end{lemma}
\begin{proof}
	Assume that the lemma is not true and let $k$ be the minimum integer such that there exist vertices $a,b,c \in \Interval{s}{t}$ such that $1 \leq k < \dist{s}{t}$ and $\{ a, b, c \} \subseteq V(\Intervalpath{s}{t})\cap \Sphere{k}{s}$. For $x\in \{a,b,c\}$, let $P_x$ and $P'_x$ denote a $(s,x)$-isometric path and $(t,x)$-isometric path, respectively. For $x\in \{a,b,c\}$, let $e_x$ and $e'_x$ denote the edge incident to $x$ in $P_x$ and $P'_x$, respectively. Let $\mathcal{Q}=\{P_a, P'_a, P_b, P'_b, P_c, P'_c\}$. Now delete from $\Intervalpath{s}{t}$ all vertices and edges that are not present in any paths in $\mathcal{Q}$. For $x\in \{a,b,c\}$, contract all edges in $E(P_x) \setminus \{e_x\}$ to $s$, and contract all edges in $E(P'_x)\setminus \{e'_x\}$ to $t$. The resulting graph is a $K_{2,3}$ with $\{s, t \}$ and $\{ a,b,c \}$ as its partite sets, a contradiction. 
\end{proof}

  \begin{lemma}\label{lem:3edge}
	For two vertices $s,t$ and an integer $i$, we define $S_i \coloneq \{xy\in E(G)\colon x\in V(\Intervalpath{s}{t})\cap \Sphere{i}{s}, y\in V(\Intervalpath{s}{t})\cap \Sphere{i+1}{s}\}$. Then,  $\cardinal{S_i}\leq 3$. 
\end{lemma}
\begin{proof}
	By \Cref{lem:2in1level}, we only need to analyse the case when both $V(\Intervalpath{s}{t})\cap \Sphere{i}{s}$ and $V(\Intervalpath{s}{t})\cap \Sphere{i+1}{s}$ have two vertices each. Let $V(\Intervalpath{s}{t})\cap \Sphere{i}{s} = \{ a,b \}$ and $V(\Intervalpath{s}{t})\cap \Sphere{i+1}{s}= \{c,d \}$. Suppose $\cardinal{S_i}=4$ and $ S_i = \{ ac, ad, bc, bd\}\subseteq E(G)$. For $x\in \{a,b\}$, let $P_x$ be an $(s,x)$-isometric path in $G$. Clearly, the subgraph formed by $S_i \cup E(P_a)\cup E(P_b)$ contains a $K_{2,3}$-minor with $\{ s, c, d \}$ and $\{a, b\}$ as partite sets, a contradiction.
\end{proof}

For two vertices $s,t\in V(G)$, if $\Intervalpath{s}{t}$ does not contain any hole, then there is a unique $(s,t)$-isometric path in $G$. In the following lemma, we shall characterise the interaction between a hole of $\Intervalpath{s}{t}$ and an $(s,t)$-isometric path in $G$.
\begin{lemma}

    Let $s,t\in V(G)$ and $C$ be a hole of $H\coloneqq \Intervalpath{s}{t}$ that does not contain $s$ or $t$ then, $V(C)$ is an $(s,t)$-cutset of $H$.
\end{lemma}
\begin{proof}
Since $C$ is a hole in $H$, there exist two vertices $u, v \in V(C)$ such that $\distG{H}{s}{u} = \distG{H}{s}{v} = i$ for some integer $i$. By \Cref{lem:2in1level}, $V(H)\cap \Sphere{i}{s} = \{ u,v \}$. Thus $H - \{ u,v \}$ is disconnected with two different connected components each containing $s$ and $t$ respectively. Thus, $H - V(C)$ is also disconnected with two different connected components.
\end{proof}

\begin{lemma}\label{lem:hole_path}
    Let $s,t\in V(G)$ and $C$ be a hole of $H\coloneqq \Intervalpath{s}{t}$ and let $P$ be an $(s,t)$-isometric path in $H$. Then $V(C) \cap V(P)$ induces a path.
\end{lemma}
\begin{proof}
    As $V(C)$ is a $(s,t)$-cutset of $H$, all $(s,t)$-isometric paths in $G$ intersect $V(C)$. Thus for any $(s,t)$-isometric path $P$, $V(P)\cap V(C) \neq \emptyset$. If the induced subgraph on $V(P)\cap V(C)$ is connected, then the lemma is true.

Let the vertices of an $(s,t)$-isometric path $P$ be naturally ordered as $v_0, v_1 \cdots, v_{\dist{s}{t}}$. Assume that the graph induced on $V(P) \cap V(C)$ has two connected components, say path $P_1 = v_i, \cdots , v_j$ and path $P_2 = v_{j+k}, \cdots , v_l$. Thus there exists an integer $j < b < j+k$ such that vertex $v_b \notin V(C)$. But by definition of a hole, there exist two vertices $x,y \in V(C) \setminus V(P)$ such that $\distG{H}{s}{x} = \distG{H}{s}{y} = \distG{H}{s}{v_b}$. This contradicts \Cref{lem:2in1level}. Thus $\InducedSub{H}{V(C) \cap V(P)}$ is a path.

\end{proof}


The above lemma implies that for any two vertices $s,t\in V(G)$, any $(s,t)$-isometric path in $G$ will contain vertices of the holes of $\Intervalpath{s}{t}$. Now we shall prove a natural ordering on the holes of $\Intervalpath{s}{t}$. For two sets $A\subseteq \Interval{s}{t}, B\subseteq \Interval{s}{t}$, we say $A\preceq B$ if for any $u\in A\setminus B$ and $v\in B\setminus A$, we have $\distG{G}{s}{u}<\distG{G}{s}{v}$. 

\begin{lemma}\label{lem:linear-order}
    Let $s,t\in V(G)$, the set of holes of $H\coloneqq \Intervalpath{s}{t}$ can be linearly ordered as $D_1,D_2,\ldots$ such that for any $i<j$ and an $(s,t)$-isometric path $P$ in $G$, $V(D_i)\cap V(P)\preceq V(D_j)\cap V(P)$.
\end{lemma}
\begin{proof}
    Since any hole of $H$ lies in some block, first, we will prove a natural ordering on the blocks of $H$. First we show that the blocks in $H\coloneqq \Intervalpath{s}{t}$ can be ordered as $B_1,B_2,\ldots$ such that for $i<j$, $V(B_i)\cap V(P)\preceq V(B_j)\cap V(P)$.

    To prove this,
	 let $a=\distG{H}{s}{t}$ and $I\subseteq \{0,1,\ldots,a\}$ be a maximal set such that for each $i\in I$, $\cardinal{N_i(s) \cap V(H)}=1$. Observe that $\{0,a\}\subseteq I$. Let the elements of $I$ be naturally ordered as $q_1, q_2, \ldots, q_p$ where $p=\cardinal{I}$. For $b\in [p-1]$, let $S_b$ denote the set of vertices $\{w\colon q_b \leq \distG{H}{s}{w} \leq q_{b+1} \}$. By \Cref{lem:same-block}, for $b\in [p-1]$, the vertices of $S_b$ induce a block $C_b$ in $H$. Clearly, for $1\leq i<j\leq p-1$, $V(C_i)\cap V(P) \preceq V(C_j)\cap V(P)$. 

    \medskip \noindent Now we will focus on holes contained in some block of $H$. We prove the following claim.

    \begin{claim}
        Let $B$ be a block of $H$. The holes of $B$ can be ordered as $C_1, C_2,\ldots$ such that for $i<j$, $V(C_i)\cap V(P)\preceq V(C_j)\cap V(P)$.
    \end{claim} 

    \begin{claimproof}
        From \Cref{lem:hole_path}, $V(C_i)\cap V(P)$ and $V(C_j)\cap V(P)$ induce two paths say, $P_1$ and $P_2$ respectively. If $V(P_1) \not\subseteq V(P_2)$ and vice versa then we are done. Assume on the contrary, $V(P_1) \subseteq V(P_2)$. As $C_i$ is a hole in $H$, there exists an integer $k$ such that there exist two vertices $u,v \in N_k(s) \cap V(C_i)$ and let $u \in N_k(s) \cap V(P_1)$. Thus $u$ is also in $N_k(s) \cap V(P_2)$. But as $C_j$ and $C_i$ are holes of $H$, $v \notin V(C_j)$. By \Cref{lem:2in1level} and \Cref{lem:same-block}, there exists a vertex $w \in V(C_j) \cap N_k(s)$. But this contradicts \Cref{lem:2in1level}.
    \end{claimproof}  

    \medskip \noindent Let $A$ and $B$ be two consecutive blocks of $H$ written as induced graph on the vertices $\{w : q\leq \distG{H}{s}{w} \leq q'\}$ and $\{w : r \leq \distG{H}{s}{w} \leq r'\}$ respectively. Clearly, $q' \leq r$. Let $C_i$ be the last hole of $A$ and $C_j$ be the first hole of $B$. Thus every vertex of $C_i$ lie in $N_k(s)$ where $k \leq q'$. Also, every vertex of $C_j$ lie in $N_k(s)$ where $k \geq r$. Thus any vertex $u \in (V(C_i) \cap V(P) ) \setminus V(C_j)$ and
    $v \in (V(C_j) \cap V(P) ) \setminus V(C_i)$, $\distG{G}{s}{u} \leq q' \leq \distG{G}{s}{v}$. Since $P$ is an $(s,t)$-isometric path, $|N_{q'}(s)| = 1$. Thus $\distG{G}{s}{u} < \distG{G}{s}{v}$. Hence $V(C_i) \cap V(P) \preceq V(C_j) \cap V(P)$.
\end{proof}

\subsection{Construction of the Auxiliary Graph $\overrightarrow{H}_{s,t}$}\label{subsec:aux}

Let $s,t \in V(G)$ and let $H \coloneq \Intervalpath{s}{t}$. Let $S_H \coloneq \{ D_1, \cdots , D_k \}$ be the linear order of holes in $H$ as realized by \Cref{lem:linear-order}. We construct a directed graph $\overrightarrow{H}_{s,t}$ as follows.

  

    \medskip \noindent \textbf{Vertex set of $\overrightarrow{H}_{s,t}$:} For a hole $D_i \in S_H$ in $H$, let $\mathcal{D}_i$ be the collection of all subpaths of the hole $D_i$ that can be traced by some $(s,t)$-isometric path in $H$. Formally, $$\mathcal{D}_i \coloneq \{ P \subseteq D_i \colon  P = \overline{P} \cap D_i \text{ for some } (s,t)\text{-isometric path }\overline{P} \text{ in }H \}.$$ Let $V_i$ be a set of vertices such that each vertex in $V_i$ represents a path in $\mathcal{D}_i$. Thus, $V_i \coloneq \{ v_P \colon P \in \mathcal{D}_i \}$. Clearly, there is a bijection between $V_i$ and $\mathcal{D}_i$. Construct vertex set of the directed graph $\overrightarrow{H}_{s,t}$ as $V(\overrightarrow{H}_{s,t}) \coloneq \displaystyle \bigcup_{i \in [k]} V_i \cup \{ s' \} \cup \{t'\}$.

    Next we prove a lemma that helps us show that $V(\overrightarrow{H}_{s,t})$ has at most $O(n)$ many vertices.

\begin{lemma}\label{lem:4-paths}
    Let $S_H \coloneq \{ D_1, \cdots , D_k \}$ be the linear order of holes in $H$. For an $i \in [k]$, let $\mathcal{D}_i \coloneq \{ P \subseteq D_i \colon  P = \overline{P} \cap D_i \text{ for some } (s,t)\text{-isometric path }\overline{P} \text{ in }H \}.$ Then $|\mathcal{D}_i| \leq 5$.
\end{lemma}
\begin{proof}
    Fix a hole $D = D_j$ in $H$ for some $j \in [k]$. Observe that for every edge $uv$ in $D$, $| \distG{H}{s}{u} - \distG{H}{s}{v}| = 1$. Also every vertex $v$ in hole $D$ has exactly two neighbors in $D$. Let $a, a' \in V(D)$ such that $\dist{s}{a} = \dist{s}{a'} = \displaystyle \min_{v \in V(D)} \dist{s}{v} = \alpha$. Thus both $a$ and $a'$ have neighbors in $D$ that lie in $N_{\alpha + 1}(s) \cap \Interval{s}{t} $. But \Cref{lem:3edge} forces a contradiction. Thus there exists a unique vertex $a \in V(D)$ such that $\dist{s}{a} = \displaystyle \min_{v \in V(D)} \dist{s}{v}$. Similarly, there exists a unique vertex $b \in V(D)$ such that $\dist{s}{b} = \displaystyle \max_{v \in V(D)} \dist{s}{v} = \beta$. Thus by \Cref{lem:2in1level} and \Cref{lem:same-block}, all vertices $v$ in $H$ are in hole $D$ if $ \alpha < \dist{s}{v} < \beta$.

    Observe that since $D$ is a hole, there are exactly two edge-disjoint $(a,b)$-isometric paths, say $L = l_0 \cdots l_h$ and $R = r_0 \cdots r_h$. Observe that $l_0 = r_0 = a$ and $l_h = r_h = b$ and $\beta - \alpha = h$. Also, $\dist{s}{l_i} = \dist{s}{r_i} = \alpha+i$ for $i \in \{ 0, \cdots ,h \}$. Let $P = \overline{P} \cap D$ be a subpath of $D$ for some $(s,t)$-isometric path $\overline{P}$. By \Cref{lem:hole_path}, $P$ is a subpath of either $L$ or $R$. 
    
    Let $P = p_1, \cdots, p_l$. Suppose $p_1 = l_i$ of $L$ where $i \geq 1$. Since $\dist{s}{l_i} = \alpha + i$, the predecessor of $l_i$ in $\overline{P}$, say $x$ lies in $N_{\alpha + i - 1}(s) \cap \Interval{s}{t}$. If $i > 1$, $\dist{s}{x} > \alpha $. Thus $x$ is in $D$. This contradicts the choice of $l_i$ as the first vertex of $P$. Therefore if $P \subseteq L$ then $i \in \{0,1 \}$. Thus $p_1 = l_0$ or $p_1 = l_1$. By similar arguments, if $P \subseteq R$ then $i \in \{0,1 \}$. Thus $p_1 = r_0$ or $p_1 = r_1$. Analogously, if $P \subseteq L$ then $p_l = l_h$ or $p_l = l_{h-1}$ and if $P \subseteq R$ then $p_l = r_h$ or $p_l = r_{h-1}$.

   Consider the vertices $a, l_1, r_1, u$ such that $\dist{s}{u} = \alpha$. By \Cref{lem:3edge}, $u$ can not have edge with both $l_1$ and $r_1$. Consider the possible first vertices of $P$. We claim that either of $r_1$ and $l_1$ is possible. If some $(s,t)$-isometric path enters $D$ through $l_1$ then $l_1$ has a predecessor $u$ in $N_{\alpha}(s) \cap \Interval{s}{t}$. Similarly, if some $(s,t)$-isometric path enters $D$ through $r_1$ then $r_1$ has a predecessor $u'$ in $N_{\alpha}(s) \cap \Interval{s}{t}$. By \Cref{lem:2in1level}, $u = u'$. But this contradicts \Cref{lem:3edge}. Thus at most one of $\{ l_1, r_1 \}$ can be the first vertex of a subpath $P$. By similar arguments at most one of $\{ l_{h-1}, r_{h-1} \}$ can be the last vertex of $P$.

   Thus maximal size set of subpaths can contain two $(a,b)$-isometric edge-disjoint paths, $(a,r_{h-1})$-isometric path, $(r_1,b)$-isometric path and $(r_1,r_{h-1})$-isometric path. Thus $|\mathcal{D}_j| \leq 5$.

\end{proof}

The bound in \Cref{lem:4-paths} is tight. \Cref{fig:tightbound} shows an example of the tight bound. From \Cref{lem:4-paths} it is clear that the size of $V(\overrightarrow{H}_{s,t})$ is $O(n)$.

\begin{figure}[ht]
    \centering
    \includegraphics[width=\linewidth]{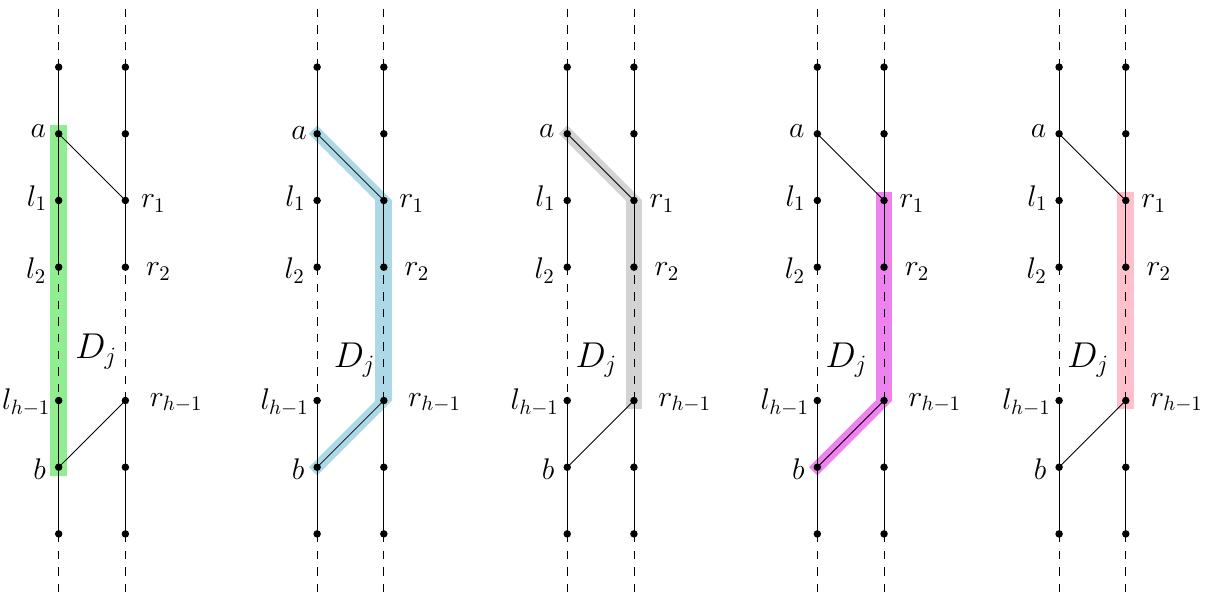}
    \caption{For a hole $D_j$ this figure lists out all possible subpaths that can be traced. Consider $\mathcal{D}_j = \{ L, R, (a,r_{h-1}) \text{-isometric path}, (r_1,b) \text{-isometric path}, (r_1,r_{h-1}) \text{-isometric path} \} $ where $L$ and $R$ are two $(a,b)$-isometric and edge-disjoint paths.}
    \label{fig:tightbound}
\end{figure}

\medskip \noindent \textbf{Edge set of $\overrightarrow{H}_{s,t}$:} Let $D_i$ and $ D_{i+1}$ be two consecutive holes in $S_H$. Assume that $V(D_i) \cap V(D_{i+1}) \neq \emptyset$ and let $v_P \in V_i$ and $v_Q \in V_{i+1}$. We add a directed edge from $v_P$ to $v_Q$, say $\overrightarrow{v_Pv_Q}$ if $P \cup Q$ is a subpath of some $(s,t)$-isometric path. If $V(D_i) \cap V(D_{i+1}) = \emptyset$, add a directed edge from $v_P$ to $v_Q$ for all $v_P \in V_i$ and for all $v_Q \in V_{i+1}$. Add a directed edge $\overrightarrow{s'v_P}$ in $E(\overrightarrow{H}_{s,t})$ for each $v_P \in V_1$ and a directed edge $\overrightarrow{v_Qt'}$ in $E(\overrightarrow{H}_{s,t})$ for each $v_Q \in V_{k}$.


Let $v_P \in V_i$ and $v_Q \in V_{i+1}$ and let $e = \overrightarrow{v_Pv_Q} \in E(\overrightarrow{H}_{s,t})$ be an edge. By construction of $\overrightarrow{v_Pv_Q}$, $P$ and $Q$ lie on a common $(s,t)$-isometric path $W$. Let $W_{PQ}$ be the minimal subpath of $W$ that contains $P$ and $Q$. Define $\Phi(\overrightarrow{v_Pv_Q}) \coloneq V(W_{PQ})$, the vertex set of $W_{PQ}$.

Next, let $v_P \in V_1$ and let $e = \overrightarrow{s'v_P}$. Define $\Phi(\overrightarrow{s'v_P}) = V(P) \cup V(P')$ where $P'$ is the shortest path in $H$ connecting $D_1$ and $s$. Note that if $P'$ is a trivial path then $\Phi(\overrightarrow{s'v_P}) = V(P)$. Similarly, let $v_Q \in V_k$ and let $e = \overrightarrow{v_Qt'}$. Define $\Phi(\overrightarrow{v_Qt'}) = V(Q) \cup V(Q')$ where $Q'$ is the shortest path in $H$ connecting $D_k$ and $t$. Note that if $Q'$ is a trivial path then $\Phi(\overrightarrow{v_Qt'}) = V(Q)$. 

We use $\Phi$ to prove the following lemma. 

\begin{lemma}\label{lem:path-correspond}
    There is a bijection between $(s',t')$-directed paths in $\overrightarrow{H}_{s,t}$ and $(s,t)$-isometric paths in $\Intervalpath{s}{t}$.
\end{lemma}
\begin{proof}
    We prove the existence of a bijection from the set of $(s',t')$-directed paths in $\overrightarrow{H}_{s,t}$ to the set of $(s,t)$-isometric paths in  $\Intervalpath{s}{t}$. By \Cref{lem:hole_path} and \Cref{lem:linear-order}, $\Intervalpath{s}{t}$ is a chain of holes $D_1, \cdots, D_k$ and unique isometric paths.

   \begin{claim}
       Every $(s,t)$-isometric path in $\Intervalpath{s}{t}$ gives a directed $(s',t')$ path in $\overrightarrow{H}_{s,t}$.
   \end{claim} 

   \begin{claimproof}
       Let $W$ be an $(s,t)$-isometric path and for each $i$, let $P_i = W \cap D_i$. By \Cref{lem:hole_path} $P_i$ is a subpath of $D_i$. The definition of $\mathcal{D}_i$ gives $P_i \in \mathcal{D}_i$. Thus $v_{P_i} \in V_{i}$ is a vertex of $\overrightarrow{H}_{s,t}$. 

        If $D_i$ and $D_{i+1}$ are disjoint then the edge $\overrightarrow{v_{P_i}v_{P_{i+1}}}$ is present by construction. If $D_i$ and $D_{i+1}$ are not disjoint then they share either a vertex or an edge, otherwise it forces a $K_{2,3}$-minor. Let $V(D_i) \cap V(D_{i+1}) \coloneq S $. If $P_i \cap P_{i+1} = S$ then $P_i \cup P_{i+1}$ is a subpath of $W$. As $W$ is isometric the edge $\overrightarrow{v_{P_i}v_{P_{i+1}}}$ is present by construction and $\Phi(\overrightarrow{v_{P_i}v_{P_{i+1}}}) = V(P_i \cup P_{i+1})$.

        Edges $\overrightarrow{s'v_{P_1}}$ and $\overrightarrow{v_{P_k}t'}$ are present by construction. 
   \end{claimproof}

\begin{claim}
    Every directed $(s',t')$ path in $\overrightarrow{H}_{s,t}$ gives an $(s,t)$-isometric path in $\Intervalpath{s}{t}$.
\end{claim}
    \begin{claimproof}
         Conversely, take a directed $(s',t')$ path in $\overrightarrow{H}_{s,t}$ with edges $e_0, e_1, \cdots, e_{k}$. Let $e_0 = \overrightarrow{s'v_{P_1}}$ and $e_{k} = \overrightarrow{v_{P_k}t'}$. For each edge $e_i$, $i \in [k-1]$, $\Phi(e_i)$ induces a subpath of an $(s,t)$-isometric path. Let $P_i \in \mathcal{D}_i$ and $P_{i+1} \in \mathcal{D}_{i+1}$. If $D_i$ and $D_{i+1}$ intersect then $\Phi(e_i) = V(P_i) \cup V(P_{i+1})$. If $D_i$ and $D_{i+1}$ are disjoint then $\Phi(e_i)$ contains the unique path joining $D_i$ and $D_{i+1}$ along with $P_i$ and $P_{i+1}$. By construction, the subpaths induced by $\Phi(e_i)$ and $\Phi(e_{i+1})$ respectively share the common subpath $P_{i+1}$. By \Cref{lem:hole_path} and \Cref{lem:linear-order}, the union $W = \displaystyle \bigcup_{e_i} \Phi(e_i)$ induces a path from $s$ to $t$ in $\Intervalpath{s}{t}$. By \Cref{lem:linear-order} and \Cref{lem:2in1level} every edge in $W$ goes from $N_j(s)$ to $N_{j+1}(s)$ for some integer $j$. Thus $W$ induces an $(s,t)$-isometric path.
    \end{claimproof} 
    
    Hence, the lemma.
\end{proof}

In the next section we use $\Phi$ to weight the edges of $\overrightarrow{H}_{s,t}$.

\subsection{The Weighting scheme}

Let $s,t \in V(G)$ and let $H \coloneq \Intervalpath{s}{t}$. Let $\overrightarrow{H}_{s,t}$ be the directed graph constructed from $H$ by the rules given in the previous section. Let $S_H \coloneq \{ D_1, \cdots , D_k \}$ be the linear order of holes in $H$ as realized by \Cref{lem:linear-order}. 

For an integer $i \in [k-1]$, consider a vertex $v_P \in V_i$ and a vertex $v_Q \in V_{i+1}$. For the edge $e = \overrightarrow{v_Pv_Q}$ in $\overrightarrow{H}_{s,t}$ we assign two weights $w_1(e)$ and $w_2(e)$. 

Consider every vertex $u \in V(G)$ that is equidistant from $\Phi(e)$ and $\Interval{s}{t}$. We define $w_1(e)$ as the maximum distance from such a vertex to $\Phi(e)$. Formally, $$w_1(e) = \displaystyle \max \bigl \{ \distG{G}{u}{\Phi(e)} \colon u \in V(G), \distG{G}{u}{\Phi(e)} = \distG{G}{u}{\Interval{s}{t}} \bigr \}.$$

Next, let $O(e) \coloneq V(D_i \cup D_{i+1}) \setminus \Phi(e)$. Consider all vertices $u$ whose shortest paths to $\Phi(e)$ all intersect $O(e)$. We define $w_2(e)$ as the maximum distance from such a vertex $u$ to $\Phi(e)$. Note that, the vertices in the set $O(e)$ are also considered in weight $w_2$. Formally, $$w_2(e) = \displaystyle \max \bigl \{ \distG{G}{u}{\Phi(e)} \colon u \in V(G),\text{ every shortest path from }u \text{ to }\Phi(e) \text{ intersects }O(e) \bigr \}.$$

Finally, the weight of the edge $\overrightarrow{v_Pv_Q}$, denoted as $\weight{\overrightarrow{v_Pv_Q}}$ is a function of $w_1(\overrightarrow{v_Pv_Q})$ and $w_2(\overrightarrow{v_Pv_Q})$. Formally, $$\weight{\overrightarrow{v_Pv_Q}} \coloneq \max \{ w_1(e) , w_2(e) \}.$$

Next we assign weights to the edges $\overrightarrow{s'v_P}$ where $v_P \in V_1$, and $\overrightarrow{v_Qt'}$ where $v_Q \in V_k$. We use the same rules which use sets $\Phi(\overrightarrow{s'v_P})$ and $\Phi(\overrightarrow{v_Qt'})$. Weight $w_1$ only depends on $\Phi$ and thus remains unchanged for both edges. For $w_2$, we define $O(\overrightarrow{s'v_P}) = V(D_1) \setminus \Phi(\overrightarrow{s'v_P})$ and $O(\overrightarrow{v_Qt'}) = V(D_k) \setminus \Phi(\overrightarrow{v_Qt'})$. Weight of the edge $e$ remains unchanged: $$\weight{e} \coloneq \max \{ w_1(e) , w_2(e) \}, $$ where $e = \overrightarrow{s'v_P}$ when $v_P \in V_1$, and $e = \overrightarrow{v_Qt'}$ when $v_Q \in V_k$.

	





In the next section we prove some structural properties of $G$ that justifies the weighting scheme presented above.

\section{Proof of Theorem 1}\label{sec:projections_minor} 



Let $G$ be a simple, connected, $K_{2,3}$-minor-free graph. For any set $S \subseteq V(G)$ and a vertex $u \in V(G) \setminus S$, define the set 
$\Projection{u}{S} = \{ v \in S \colon \dist{u}{v} = \dist{u}{S} \}$. The set $\Projection{u}{S}$ is the set of \emph{projections} of $u$ onto $S$.

\begin{lemma}\label{lem:pro1}
    For $s,t \in V(G)$ and $u \in V(G)\setminus \Interval{s}{t}$, there exists an $(s,t)$-isometric path $P$ such that $\Projection{u}{\Interval{s}{t}} \subseteq V(P)$. 
\end{lemma}
\begin{proof}
    If $\Projection{u}{\Interval{s}{t}}$ is a singleton set, then we are done. Assume that there exists no $(s,t)$-isometric path $S$ such that $\Projection{u}{\Interval{s}{t}} \subseteq V(S)$. Thus there exist $a,b \in \Projection{u}{\Interval{s}{t}}$ and $(s,t)$-isometric paths $P$ and $Q$ such that $a \in V(P) \setminus V(Q)$ and $b \in V(Q) \setminus V(P)$. There exists a $K_{2,3}$-minor with $\{ s,t,u \}$ and $\{ a,b \}$ as two partite sets.
\end{proof}



\begin{lemma}\label{lem:pro2}
    For $s,t \in V(G)$ and $u \in V(G)\setminus \Interval{s}{t}$, let $\Projection{u}{\Interval{s}{t}} \subseteq V(C)$ where $C$ is a block in $\Intervalpath{s}{t}$. Then $1 \leq \cardinal{\Projection{u}{\Interval{s}{t}}} \leq 2$. Moreover if $\Projection{u}{\Interval{s}{t}} = \{a,b \}$ then $\dist{a}{b} = 1$. 
\end{lemma}
\begin{proof}
    As the graph is connected, $\cardinal{\Projection{u}{\Interval{s}{t}}} \geq 1$. Assume on the contrary $\cardinal{\Projection{u}{\Interval{s}{t}}} \geq 3$ and $\{a, b, c \} \subseteq \Projection{u}{\Interval{s}{t}}$. Without loss of generality, let $\dist{s}{a} < \dist{s}{b} < \dist{s}{c}$. Consider a vertex $a_1 \in (u,a)$-path. Contract the edges in this path such that there is an edge between $u$, $a_1$ and an edge between $a_1$, $a$. Similarly consider vertices $b_1 \in (u, b)$-path and $c_1 \in (u, c)$-path. Contract edges in these paths such that there is an edge between $u$, $b_1$ and an edge between $b_1$, $b$ and there is an edge between $u$, $c_1$ and an edge between $c_1$, $c$. Contract all the edges in a $(a,c)$-path to the vertex $a$. Thus the induced subgraph on the vertices $\{ u, a_1, b_1, c_1, a \}$ forms a $K_{2,3}$-minor in $\Intervalpath{s}{t}$. This is a contradiction. Else if $ua, ub,uc \in E(G)$, then $u \in \Interval{s}{t}$. But this contradicts Lemma \ref{lem:2in1level}.

    Let $\Projection{u}{\Interval{s}{t}} = \{ a,b \}$ and $\dist{a}{b} > 1$. Let $c$ be a vertex in an $(a,b)$-isometric path. Let $c \in V(\Intervalpath{s}{t})\cap \Sphere{i}{s}$ for some $i$. As $a,b,c \in V(C)$, there exists a vertex $d \in V(C)$ such that $d \in V(\Intervalpath{s}{t})\cap \Sphere{i}{s}$. There exists a $K_{2,3}$-minor with $\{ u, c, d \}$ and $\{ a, b \}$ as partite sets. Thus contradiction.
\end{proof}

We have the following observation:

\begin{observation}\label{obs:pro3}
    For $u,s,t \in V(G)$, if a projection $\Projection{u}{\Interval{s}{t}}$ intersects an isometric path connecting any two consecutive blocks of $\Intervalpath{s}{t}$ then, each $(s,t)$-isometric path intersects $\Projection{u}{\Interval{s}{t}}$.
\end{observation}



\begin{lemma}\label{lem:pro4}
     Let $u,s,t \in V(G)$ and assume the set $\Projection{u}{\Interval{s}{t}}$ has no cut-vertex. Let $C_i$ and $C_j$ be two blocks in $\Intervalpath{s}{t}$ such that $i<j$ and $\Projection{u}{\Interval{s}{t}}$ intersects $C_i$ and $C_j$. The following statements are true$\colon$

     \begin{enumerate}
         \item There exists no block $C_k$ in $\Intervalpath{s}{t}$ such that $k \neq i$ and $k\neq j$ and $\Projection{u}{\Interval{s}{t}}$ intersects $C_k$. Moreover $j=i+1$.
         
         \item For any two vertices $a,b \in \Projection{u}{\Interval{s}{t}}$, $\dist{a}{b} \leq \delta +2$ where $\delta$ is the length of the isometric path connecting $C_i$ and $C_j$.
         \item There does not exist any vertex $v \in V \setminus \Interval{s}{t}$ with $v \neq u$ such that $\Projection{v}{\Interval{s}{t}}$ intersects both $ V(C_i) $ and $ V(C_{j})$.
     \end{enumerate}
      
\end{lemma}

\begin{proof}
    \begin{enumerate}
        
        \item Assume on the contrary that there exists a block $C_k$ such that $\Projection{u}{\Interval{s}{t}}$ intersects $C_k$. Without loss of generality, let $i<k<j$. Let $a \in \Projection{u}{\Interval{s}{t}} \cap V(C_i)$, $b \in \Projection{u}{\Interval{s}{t}} \cap V(C_k)$ and $c \in \Projection{u}{\Interval{s}{t}} \cap V(C_j)$. As $b$ is not a cut-vertex, there exists a vertex $v \in V(C_k)$ such that $\dist{s}{b} = \dist{s}{v}$. This forces a $K_{2,3}$ minor with $\{u,v \}$ and $\{ a,b,c \}$ as two partite sets.

        Assume $j > i+1$. Thus there exists a block $C_k$ such that $i < k < j$. Let $a \in \Projection{u}{\Interval{s}{t}} \cap V(C_i)$ and $b \in \Projection{u}{\Interval{s}{t}} \cap V(C_j)$. Let two vertices $y,z \in V(C_k)$ such that $\dist{s}{y} = \dist{s}{z}$. This forces a $K_{2,3}$-minor with $\{a,b\}$ and $\{u,y,z\}$ as two partite sets.

        \item Assume there exists two vertices $a,b \in \Projection{u}{\Interval{s}{t}}$, $\dist{a}{b} > \delta + 2$ where $\delta$ is the length of the isometric path connecting $C_i$ and $C_j$. Let this path be $P$. If $\dist{a}{P} =1$ and $\dist{b}{P} =1$ we are done. Assume without loss of generality $\dist{a}{P}>1$. Let $y,z \in V(C_i)$ such that $\dist{y}{P} = \dist{z}{P} = 1$. This forces a $K_{2,3}$-minor with $\{u,y,z \}$ and $\{ a,b \}$ as two partite sets.

        \item Assume on the contrary that there exist two vertices $u,v \in V \setminus \Interval{s}{t}$ such that $\Projection{u}{\Interval{s}{t}}$ and $\Projection{v}{\Interval{s}{t}}$ intersect both $ V(C_i) $ and $ V(C_{i+1})$. Let $p \in V(C_i)$ such that $\dist{s}{p} = \dist{s}{V(C_i)}$ and $q \in V(C_{i+1})$ such that $\dist{t}{q} = \dist{t}{V(C_{i+1})}$. If $z$ is a vertex in the unique isometric path connecting $C_i$ and $C_{i+1}$, then there is a $K_{2,3}$-minor with $\{ u,v,z \}$ and $\{ p,q \}$ as two partite sets. Else if $V(C_i) \cap V(C_{i+1}) = \{z \}$ for some vertex $z$ then, there is a $K_{2,3}$-minor with $\{ u,v,z \}$ and $\{ p,q \}$ as two partite sets. Thus contradiction.
    \end{enumerate}
\end{proof}

\begin{lemma}\label{lem:weight-correspond}
    If the maximum edge-weight of a directed $(s',t')$ path in $\overrightarrow{H}_{s,t}$ is $k$ then the eccentricity of its corresponding $(s,t)$-isometric path in $\Intervalpath{s}{t}$ is $k$.
\end{lemma}
\begin{proof}
    Let $P_{s',t'}$ be a directed $(s',t')$ path in $\overrightarrow{H}_{s,t}$ with maximum edge-weight $k$. Let $Q_{s,t}$ be the corresponding $(s,t)$-isometric path in $\Intervalpath{s}{t}$. We prove $\ecc{Q_{s,t}}=k$. For an edge $e$ of $P_{s',t'}$, let $\weight{e} =k$. Since $\weight{e} = \max \{ w_1(e), w_2(e) \}$, either $w_1(e) = k$ or $w_2(e) = k$.

    \begin{claim}\label{claim:cor-1}
        $k \leq \ecc{Q_{s,t}}$.
    \end{claim} 

    \begin{claimproof}
         Assume $w_1(e) = k$. Thus there exists a vertex $u \in V(G) \setminus \Interval{s}{t}$ such that $\dist{u}{\Phi(e)} = \dist{u}{\Interval{s}{t}} = k$. As $\Phi(e) \subseteq V(Q_{s,t}) \subseteq \Interval{s}{t}$, $\dist{u}{Q_{s,t}} = k$, so $k \leq \ecc{Q_{s,t}}$.
         
         \medskip \noindent  Next assume $w_2(e) = k$.
         Let $u \in V(G)$. Consider the set $\Projection{u}{\Interval{s}{t}}$. By \Cref{lem:pro2} and \Cref{lem:pro4}, this set intersects at most two consecutive holes in $\Intervalpath{s}{t}$. If the two consecutive holes $D_i$ and $D_{i+1}$ corresponding to $e$, are disjoint then $\Phi(e)$ contains the vertices of the unique path connecting $D_i$ and $D_{i+1}$. Thus by \Cref{obs:pro3} and definition of $w_1(e)$, $\Projection{u}{\Interval{s}{t}}$ does not have any intersection with this path. If $D_i$ and $D_{i+1}$ intersect then by definition of $O(e)$, $\Projection{u}{\Interval{s}{t}} \subseteq O(e)$. If $e$ is the first edge in $P_{s',t'}$, $\Phi(e)$ contains the vertices of the path connecting $s$ and $D_1$. Similar arguments follow if $e$ is the last edge of $P_{s',t'}$. Thus $\Projection{u}{\Interval{s}{t}} \subseteq O(e)$. 
    Thus, $\dist{u}{Q_{s,t}} = k$, so $k \leq \ecc{Q_{s,t}}$.
    \end{claimproof}


    \begin{claim}\label{claim:cor-2}
        $k \geq \ecc{Q_{s,t}}$.
    \end{claim} 

    \begin{claimproof}
         Let $u \in V(G)$ such that $\ecc{Q_{s,t}} = \dist{u}{Q_{s,t}}$. We find an edge $e$ in $\overrightarrow{H}_{s,t}$ with $\weight{e} \geq \dist{u}{Q_{s,t}}$. If $u \in \Interval{s}{t}$ and $u \notin V(Q_{s,t})$ then $u$ is in $O(e)$ for some $e$. Thus $\dist{u}{Q_{s,t}} = \dist{u}{\Phi(e)}$. Also, $w_2(e) \geq \dist{u}{Q_{s,t}}$. Thus $k \geq \dist{u}{Q_{s,t}} = \ecc{Q_{s,t}}$.

    \medskip \noindent If $u \notin \Interval{s}{t}$, consider $\Projection{u}{\Interval{s}{t}}$. If $\Projection{u}{\Interval{s}{t}}$ intersects $V(Q_{s,t})$ then for some $e$ in $ \overrightarrow{H}_{s,t}$ $u$ has been considered in $w_1(e)$. Else if $\Projection{u}{\Interval{s}{t}}$ does not intersect $V(Q_{s,t})$ there exists an edge $e$ in $\overrightarrow{H}_{s,t}$ such that $\Projection{u}{\Interval{s}{t}} \subseteq O(e)$. Thus $u$ is considered in $w_2(e)$. Thus $w_2(e) = \dist{u}{\Phi(e)} = \dist{u}{Q_{s,t}}$. This yields that some edge $e'$ in $\overrightarrow{H}_{s,t}$ has $\weight{e'} \geq \ecc{Q_{s,t}}$.
    \end{claimproof} 

    Hence from \Cref{claim:cor-1} and \Cref{claim:cor-2}, the lemma is proved.
\end{proof}

\noindent A $K_{2,3}$-minor-free graph of $n$ vertices has $O(n)$ edges. As a one-time preprocessing step, we compute all-pairs shortest distances by running a BFS from each vertex, which takes $O(n \times m) = O(n^2)$ time in total and is dominated by the overall running time. For a pair $s,t \in V(G)$, finding $\Interval{s}{t} = \{ v \colon \dist{s}{v} + \dist{v}{t} = \dist{s}{t} \}$ and the edges in $\Intervalpath{s}{t}$, say  $E(\Intervalpath{s}{t}) = \{ xy \colon \dist{s}{x} + \dist{y}{t} + 1 = \dist{s}{t} \}$ take $O(n)$-time. By ordering the holes in $\Intervalpath{s}{t}$ by their distance to $s$, we find out the order realized by \Cref{lem:linear-order} in $O(n)$-time. As each $V_i$ has at most $5$ vertices (see \Cref{lem:4-paths}), tracing out the subpaths in $\mathcal{D}_i$ is local and thus takes a total of $O(n)$-time. We use standard text book algorithms to find a $(s',t')$-path in $\overrightarrow{H}_{s,t}$ that minimizes the maximum edge-weight in the path. This can be done using dynamic programming in $O(n)$-time. Performing one BFS from each interval vertex computes the projections. This takes $O(n \times \cardinal{\Interval{s}{t}}) = O(n^2)$ time. Computing $w_1$ for each edge $e$ takes $O(n)$-time as it is $\distG{G}{u}{\Interval{s}{t}}$ where $u$ is a farthest vertex from $\Phi(e)$. Weight $w_2(e)$ also takes $O(n)$-time after all-pair-shortest-distances have been computed. Thus running this subroutine over all $\binom{n}{2}$ pairs of vertices, gives the running time $O(n^2 \times n^2) = O(n^4)$. Thus we have the following result.

 \main*

\section{Proof of Theorem 2}\label{sec:cactus}
We list some lemmas on the \emph{projections} on an \emph{interval} in a cactus graph. For the rest of the section, we consider a graph $G$ as a cactus graph. Since no pair of cycles in a cactus share an edge, we have the following observation: 


\begin{observation}\label{obs:cac1}
    Let $s,t \in V(G)$ and let $H \coloneq \Intervalpath{s}{t}$. Every block $C$ in $H$ is a hole.
\end{observation}




\begin{lemma}\label{lem:cac2}
    For vertices $s,t \in V(G)$ and a vertex $u \in V(G)\setminus \Interval{s}{t}$.
    \begin{enumerate}
        \item $|\Projection{u}{\Interval{s}{t}}| \leq 2$.
        \item If $\Projection{u}{\Interval{s}{t}} = \{ a,b \}$ then $a$ and $b$ are cut vertices in $\Intervalpath{s}{t}$. Consequently, if for any hole $C$, $\Projection{u}{\Interval{s}{t}} \cap C \neq \emptyset$ then $|\Projection{u}{\Interval{s}{t}}| = 1 $.
    \end{enumerate}
\end{lemma}

\begin{proof}
    \begin{enumerate}
        \item Without loss of generality, assume $\Projection{u}{\Interval{s}{t}} = \{ a,b,c \}$. Let $P$, $Q$ and $R$ be any $(u,a)$-path, $(u,b)$-path and $(u,c)$-path respectively. Note that $P$,$Q$ and $R$ are part of cycles which share an edge. Thus contradiction.
        \item Assume on the contrary that not both of $a$ and $b$ are cut vertices. Without loss of generality, let $a \in V(C)$ where $C$ is a hole in $\Intervalpath{s}{t}$. Thus, at least one edge in the $(a,b)$-path, say $e$ is also an edge of $C$. Also, a $(u,a)$-path, the $(a,b)$-path and a $(b,u)$-path form a cycle. Thus $e$ is also part of this cycle. Thus contradiction.
    \end{enumerate}
\end{proof}

As a cactus is a $K_{2,3}$-minor-free graph, for fixed vertices $s,t$ in $G$ the subroutine presented in \Cref{sec:algorithm} returns an $(s,t)$-isometric path with minimum eccentricity. We use \Cref{obs:cac1} and \Cref{lem:cac2} to derive the time needed by the algorithm. In the algorithm the step that takes $O(n^2)$-time for $K_{2,3}$-minor-free graphs, is the computation of the projections of vertices in $G - \Interval{s}{t}$. In a cactus, a connected component of $G - \Interval{s}{t}$ is connected to $\Interval{s}{t}$ via one hole or a unique subpath joining two consecutive holes (see \Cref{lem:cac2}). A BFS confined to a connected component computes its farthest vertex from $\Interval{s}{t}$. Since the components are disjoint, each BFS traversal sums up to $O(n)$-time. Thus finding an $(s,t)$-isometric path with minimum eccentricity takes $O(n)$-time. Consequently in cactus, MESP admits an $O(n^3)$-time algorithm. Thus we have the following result.

\mainn*

\section{Conclusion}\label{sec:conc}
In this paper, we prove that \textsc{MESP} can be solved in polynomial time on $K_{2,3}$-minor-free graphs, and in cubic time on cactus graphs. The main ingredient of our algorithm is a structural analysis of the interval $F_{s,t}$ between two fixed vertices $s$ and $t$. Although this interval may contain exponentially many $(s,t)$-isometric paths, the exclusion of a $K_{2,3}$ minor forces a strong linear structure on the holes of $F_{s,t}$. Also each shortest path interacts with every hole in only a constant number of possible ways, and the projections of vertices outside the interval are highly restricted. This allows us to replace the family of all $(s,t)$-isometric paths by a linear-size auxiliary directed graph.

We believe that this approach is useful beyond the graph classes considered here. A common difficulty in \textsc{MESP} is that the set of shortest paths between two vertices can be very large, while the objective depends globally on the distance of all vertices to the chosen path. Our method separates these two issues. In particular, it would be interesting to determine whether similar decompositions to the ones in this paper exist for larger subclasses of planar graphs. For example series-parallel graphs or other bounded-treewidth graph classes where the complexity of \textsc{MESP} is still open.






    

\bibliographystyle{plain}
\bibliography{lipics-v2021-sample-article}

\appendix

\end{document}